\documentclass[12pt,oneside,reqno]
{amsart}
\usepackage[T1]{fontenc}
\usepackage[utf8]{inputenc}
\usepackage[english]{babel}

\usepackage{amsmath,amssymb,mathtools,bbm}
\usepackage{enumerate}
\usepackage{array}
\usepackage{booktabs}
\usepackage{caption}

\usepackage[labelfont=md]{subcaption}

\usepackage[
  a4paper,
  tmargin=1in,
  bmargin=1.2in,
  lmargin=1in,
  rmargin=1in,
  headheight=0in,
  headsep=0in,
  footskip=1.5\baselineskip
]{geometry}

\usepackage{natbib}
\usepackage{nicefrac}
\usepackage{setspace}
\usepackage{tabularx}
\usepackage{graphicx}
\usepackage{tikz}
\usepackage{float}
\usepackage{paralist}
\usepackage{calc}
\usepackage{pgfplots}
\usepackage{adjustbox}
\pgfplotsset{compat=1.7}
\usepackage{verbatim}

\numberwithin{equation}{section}

\usepackage[colorlinks=true,linkcolor=blue,citecolor=blue]{hyperref}

\let\oldFootnote\footnote
\newcommand\nextToken\relax
\renewcommand\footnote[1]{%
    \oldFootnote{#1}\futurelet\nextToken\isFootnote}
\newcommand\isFootnote{%
    \ifx\footnote\nextToken\textsuperscript{,}\fi}

\usepackage{amsthm}
\newtheorem{theorem}{Theorem}
\newtheorem*{theorem*}{Theorem}
\newtheorem{lemma}{Lemma}
\newtheorem{proposition}{Proposition}
\newtheorem{claim}{Claim}

\theoremstyle{definition}

\newtheorem*{definition*}{Definition}
\newtheorem*{lemma*}{Lemma}

\newcommand{\eps}{\varepsilon}

\newcommand{\E}{\mathbb{E}}
\newcommand{\Prob}{\mathbb{P}}
\DeclareMathOperator*{\argmax}{arg\,max}

\def\ee{\mathrm{e}}

\title{Action Volatility in a Changing World}

\author{Wanying Huang}
\thanks{Department of Economics, Monash University. Email: \texttt{kate.huang@monash.edu}\\
\emph{Acknowledgements.} This work is based on the third chapter of my PhD dissertation at Caltech and was previously circulated with the title ``The emergence of fads in a changing world''.\\ 
I thank the editor, an associate editor, and three anonymous referees for their constructive comments and
suggestions, which substantially improved the paper. I am indebted to Omer Tamuz for his continued support and encouragement. I would like to thank Arthur Campbell, Krishna Dasaratha, Keaton Ellis, Sumit Goel, Wade Hann-Caruthers, Pawel Janas, Aldo Lucia, Thomas Palfrey, Luciano Pomatto and Jean-Laurent Rosenthal for their helpful comments. All errors are my own.}

\date{\today}

\begin{document}

\begin{abstract}
We study how volatile agents’ behavior is in a social learning environment with a changing state. We consider a simple sequential social learning model where rational agents arrive in order, each acting only once, and the underlying unknown state constantly evolves. Each agent receives a private signal, observes all past actions of others, and chooses an action to match the current state. Because the state changes over time, cascades cannot last forever, and actions also fluctuate. We show that actions change less often than private signals but, despite the presence of temporary information cascades, still more often than the underlying state in the long run. 
\end{abstract}

\maketitle









\section{Introduction}
In many economic environments, individuals rely on the observed decisions of others to guide their own choices. At the same time, the conditions they are learning about\textemdash such as market demand, the value of new technologies, or organizational constraints\textemdash often change over time. In these situations, learning is not simply a matter of identifying a fixed state of the world, but of adjusting actions as new information arrives and circumstances evolve. This raises a natural question: does behavior adjust more or less frequently than the underlying conditions themselves? Understanding this relative frequency of adjustment is important in many domains in which behavioral volatility, or the lack thereof, can contribute to recurring mismatches and potentially decouple market outcomes from underlying fundamentals. 

As an illustrative example, consider a university student choosing a field to major in. In addition to her own private research, she can infer future labor-market prospects, such as wages and working conditions, by observing the choices of recent graduates. However, this social information is only partially relevant, since the relative attractiveness of different fields evolves over time. As a result, the popularity of certain majors may fluctuate, and these fluctuations may not perfectly reflect underlying changes in market demand, potentially creating a wedge between labor supply and demand. Similarly, in financial markets, when professional traders decide whether to buy or sell a stock, they often rely on past trading volumes to infer an asset's intrinsic value, which itself changes over time. Consequently, trading behavior that adjusts too slowly or too quickly to fundamentals can generate price deviations from intrinsic value, giving rise to speculative bubbles and fads \citep{camerer1989bubbles, aggarwal1990fads}.

To understand the extent to which actions fluctuate in these settings, we consider a canonical sequential social learning model in which agents aim to match a binary state. Unlike the standard fixed-state framework \citep{banerjee1992simple, BichHirshWelch:92}, the state evolves over time: in each period, it switches with a small, symmetric probability.
Agents arrive sequentially, each acting once and receiving a positive payoff if her action matches the current state. Before making a decision, each agent receives a private binary signal that is informative about the current state and observes the past actions of her predecessors. As the underlying state evolves, the optimal action fluctuates as well. The question we ask is: how frequently do actions change, and more specifically, how does their long-run volatility compare to that of the underlying state and the private signals that agents receive?




We focus on a setting in which the underlying state is relatively stable, so that the information inferred from past actions remains relevant and accumulates over time. In this setting, \cite*{moscarini1998social} show that information cascades still arise, but they can only be temporary. As in a static environment, whenever a cascade arises, social information overwhelms agents’ private signals, and so agents mimic their predecessors even when their own signals suggest otherwise. However, as the state evolves, the social information sustaining a cascade becomes progressively less relevant to the current agent. Once this information has depreciated enough, the cascade breaks and agents resume responding to their private signals. Hence, in a dynamic environment, cascades are temporary---this is in sharp contrast to cascades in a static environment, which can persist indefinitely once formed.


Given the presence of information cascades, one might expect actions to fluctuate less frequently than signals, but the comparison with the underlying state is \textit{a priori} unclear. On the one hand, because agents ignore their private signals during cascades, they fail to adjust their actions when the state changes.\footnote{The symmetry of the binary state further amplifies this effect. For example, suppose the state changes
an even number of times, say twice. An agent in a cascade may then incorrectly infer that the state is unchanged and thus have no reason to change her action.} On the other hand, since cascades are temporary and signals are imperfect, agents may also switch actions even when the state remains unchanged. 

Our main result (Theorem \ref{cor:exptime}) answers this question by establishing a strict ordering of state, action, and signal volatility. In particular, we show that, despite the presence of temporary information cascades, actions fluctuate strictly more frequently than the underlying state but strictly less frequently than private signals in the long run. 
In other words, while information cascades reduce responsiveness to noisy signals, they do not eliminate excess action volatility relative to the state. It is worth emphasizing that agents are fully rational in our model, and this excess action switching relative to the state arises from their incentives to match the evolving state given their information, rather than from payoff externalities or heuristic behavior. 

For example, consider a regime in which the state changes once every 100 periods on average and each private signal matches the current state with probability 0.8. In this case, signals change approximately once every three periods. We show that the expected time between action changes is strictly greater than three periods but strictly less than 44 periods, which is well below the average duration between state changes. As a result, in the long run, actions switch less often than signals but still more often than the state.

The idea behind the proof of Theorem \ref{cor:exptime} is as follows. Intuitively, when agents are in a cascade, they simply mimic the actions of their predecessors, which dampens action volatility relative to signal volatility. However, this effect is limited by the bounded duration of cascades. Once a cascade ends, agents again respond to their private signals and are therefore susceptible to noisy or conflicting signals. These recurrent episodes of renewed responsiveness then generate excess action volatility relative to the state. Formally, we establish these comparisons by upper bounding the expected times between cascade entries and action switches. The bound on the time between cascade entries implies that cascades occur with positive frequency in the long run, thereby reducing action volatility relative to signals. Meanwhile, we show that the bound on action switches is strictly lower than the expected time between state changes, implying that action switches occur more frequently than state changes on average. Together, these results imply that the long-run action volatility is lower than that of signals but higher than that of the state.

Finally, we examine how access to more information affects action volatility by considering an alternative setting in which agents observe all past signals rather than actions. We identify regimes in which the additional information leads agents to switch actions more frequently. In other environments, however, observing signals reduces action volatility despite the absence of information cascades. These contrasting results thus reflect the tension between the advantage of observing more information and the stabilizing effect of information cascades, with either effect dominating depending on signal precision and state volatility.

\subsection*{Related Literature.} \label{sec:lit}
Despite an extensive literature on social learning,\footnote{See \cite{bikhchandani2024information} for a comprehensive survey of recent developments in the social learning literature.} only a small strand of papers considers environments with a changing state. Nevertheless, in the pioneering work of \cite{BichHirshWelch:92}, the idea of a changing state is briefly mentioned in their discussion of the fragility of information cascades. Through a numerical example, they show how a one-time shock to the state could break a cascade, even if such a shock is never realized, and that the probability of an action change is at least 87\% higher than the probability of a state change (see their Result 4), which aligns with our main result. They describe this phenomenon as ``\emph{seemingly whimsical swings in mass behavior without obvious external stimulus}'' and refer to it as fads. 

While \cite{BichHirshWelch:92} presents an early notion of fads, they mainly focus on learning in a fixed environment where fads cannot recur indefinitely. In contrast, the recurrence of fads is possible in a changing environment. Indeed, as shown later in \citet*{moscarini1998social}, if the underlying state evolves in every period but remains relatively stable, information cascades must arise, though they can only be temporary, i.e., each must end in finite time. Our work builds on their model and asks the following natural question: how does the long-run volatility of actions compare to that of signals and the underlying state? 


In a setting where a single agent repeatedly receives private signals, \citet*{hirshleifer2002economic} examine the effect of \emph{memory loss}\textemdash a situation in which the agent only recalls past actions but not past signals\textemdash on the continuity of the agent's behavior. They analyze the equilibrium behavior of a
five-period stylized model and show that in a relatively stable environment, memory loss induces excessive action inertia compared to a full-recall regime. In contrast, in a more volatile environment, memory loss results in excessive action impulsiveness.\footnote{Intuitively, as volatility of the environment increases, past actions become less relevant to the current state. At some point, this information weakens enough so that the amnesiac agent would always follow her latest signal, but the full-recall agent may not do so at this point. Hence, there is an increase in the probability of an action change due to amnesia.} Different from their work\textemdash which examines how an agent’s ability to recall past actions and signals affects action fluctuations in the short-term\textemdash our study investigates the long-term fluctuations in the actions of agents who observe past actions in a changing environment.

Among a few more recent studies that consider a dynamic state, the efficiency of learning has been a primary focus of study. For example, \citet*{frongillo2011social} consider a specific environment in which the underlying state follows a random walk with non-Bayesian agents who use different linear rules when updating. Their main result is that the equilibrium updating weights may be Pareto suboptimal, causing inefficiency in learning.\footnote{See more studies in the computer science literature, e.g., \citet*{acemoglu2008convergence, shahrampour2013online} that consider a dynamic environment with non-Bayesian agents.} In a similar but more general environment, \citet*{dasaratha2023learning} 
show that having sufficiently diverse network neighbors with different signal distributions improves learning efficiency, as such diversity allows agents to extract the most relevant information from old and confounded data.

In a setup similar to ours, a recent study by \citet*{levy2024stationary} considers the welfare implications of a dynamic state. In their model, agents observe a random subsample drawn from all past actions and then decide whether to acquire private signals that are potentially costly. These model generalizations allow them to highlight the trade-off between learning efficiency and responsiveness to environmental changes in maximizing equilibrium welfare. In contrast, in our main setting, we assume that agents observe the full history of past actions, and there is no cost associated with obtaining their private signals. We focus on this canonical sequential learning model without further complications as our objective is to compare the long-term relative frequency of action, signal, and state changes\textemdash a question that turns out to be nontrivial even in this simple setting.

\section{Model}
We follow the setup from \citet*{moscarini1998social} closely. Time is discrete, and the horizon is infinite, i.e., $t \in \mathbbm{N}_+ = \{1, 2, \ldots\}$. There is a binary state $\theta_t \in \{+1, -1\}$ that evolves over time. Specifically, the state evolves according to a Markov chain with a symmetric transition probability $\eps \in (0,1)$:
$$\Prob[\theta _{t+1} \neq i \mid \theta_t = i]= \eps, \text{~for~} i \in \{+1, -1\}.$$
For simplicity, we assume that both states are equally likely at the beginning of time. Note that this uniform distribution is also the stationary distribution of $\theta_t$.

A sequence of short-lived agents indexed by time $t$ arrive in order, each acting once by choosing an action $a_t \in \{+1, -1\}$. For each agent $t$, she obtains a payoff of one if her action matches the current state, i.e., $a_t = \theta_t$ and zero otherwise. Before choosing an action, she receives a binary private signal $s_t \in \{+1, -1\}$ and observes the history of all past actions made by her predecessors, $h_{t-1} := (a_1, \dots, a_{t-1})$. Conditional on the entire sequence of states, the private signals $(s_t)$ are independent, and each $s_t$ follows a Bernoulli distribution $B_{\theta_t}(\alpha)$ where $\alpha \in (1/2, 1)$ is the probability of matching the current state: 
$$\Prob[s_t = i \mid  \theta_t = i] = \alpha, \text{~for~} i \in\{+1, -1\}.$$

At any time $t$, the sequence of events is as follows. First, the agent arrives and observes the history of all past actions, $h_{t-1}$. Second, the state $\theta_{t-1}$ transitions to $\theta_t$ with a probability $\eps$ of switching. After the state transitions, agent $t$ then receives a private signal $s_t$ that matches the current state $\theta_t$ with probability $\alpha$. Finally, she chooses an action $a_t$ that maximizes the probability of matching $\theta_t$, conditional on the information available to her, $(h_{t-1}, s_t)$.

Throughout, we assume that the state is \emph{relatively stable}: for any signal precision $\alpha\in (1/2, 1)$, the probability of a state change is $\eps \in (0, \alpha(1-\alpha))$. Equivalently, one can think of this assumption as follows: in every period, with probability $2\eps \in (0, 2\alpha(1-\alpha))$, the state is redrawn from the set $\{+1, -1\}$ with equal probability. Thus, the probability of a state change is equal to $\eps \in (0, \alpha(1-\alpha))$. As shown in \cite{moscarini1998social}, this assumption ensures information cascades can arise but only temporarily. 


\subsection*{Agents' Behavior.} \label{sec:belief dynamics}
Let $p_t:= \Prob[\theta_t = +1 \mid h_{t-1}, s_t]$ denote the \textit{posterior belief} of agent $t$ that the state is positive after observing the pair consisting of the action history and private signal $(h_{t-1}, s_t)$. The log-likelihood ratio (LLR) of agent $t$'s posterior belief that the state is $+1$ relative to the state being $-1$ is 
\[
L_t := \log \frac{p_t}{1-p_t} = \log \frac{\Prob[\theta_t = + 1\mid h_{t-1}, s_t]}{\Prob[\theta_t =-1\mid h_{t-1}, s_t]}.
\] 
We call $L_t$ the \emph{posterior LLR} at time $t$. By Bayes' rule, it can be written as the sum of two terms:
\begin{align*} 
L_t = \log \frac{\Prob[\theta_{t} = + 1\mid  h_{t-1}]}{\Prob[\theta_{t} = -1\mid  h_{t-1}]}  + \log \frac{\Prob[s_t \mid  \theta_{t} = + 1, h_{t-1} ]}{\Prob[s_t \mid \theta_{t} = -1,  h_{t-1}]}. 
\end{align*}

We refer to the first term as the \emph{public LLR} at time $t$ and denote it by $\ell_t$. This is the log-likelihood ratio of agent $t$’s \emph{public belief}, which is derived from the history of past actions $h_{t-1}$ and is given by $q_t := \Prob[\theta_t = +1 \mid h_{t-1}]$. Since the private signal is conditionally independent of past actions given the current state, the second term in the above sum simplifies to the LLR induced by the signal alone. We define this value as $c_\alpha := \log \frac{\alpha}{1-\alpha}$. Thus, the second term in the above sum equals $c_\alpha$ if $s_t = +1$ and $-c_\alpha$ if $s_t =-1$. Consequently, depending on the realization of the private signal, the evolution of the posterior LLR can be written as follows:
\begin{align} \label{eq:publicLLR} 
L_t = \ell_t + c_\alpha \cdot s_t
    \end{align}

The optimal action for agent $t$ is the action that maximizes her expected payoff conditional on the information available to her:
\begin{align*}
    a_t & \in \argmax_{a\in \{-1, +1\}} \Prob[\theta_t= a\mid h_{t-1}, s_t].
\end{align*} 
It thus follows from \eqref{eq:publicLLR} that $a_t = +1$ if $L_t >0$ and $a_t = -1$ if $L_t <0$. When $L_t = 0$, agent $t$ is indifferent between both actions. We assume that in this case she follows the action taken by her immediate predecessor, i.e., $a_t =a_{t-1}$.~
This tie-breaking rule differs from that used in \citet*{moscarini1998social}, where indifferent agents are assumed to follow their own private signals. We make this assumption so that any action changes are driven by agents' strict preference for one action over another, rather than by the specific choice of the tie-breaking rule.

\subsection*{Cascade and Learning Regions.} As mentioned before, an \textit{information cascade} occurs when the social information inferred from past actions outweighs an agent's private signal, causing the agent to disregard her private information. From \eqref{eq:publicLLR}, one can see that this occurs precisely when $\lvert \ell_t\rvert > c_\alpha$. In this case, the sign of $L_t$\textemdash and therefore the optimal action for agent $t$\textemdash is determined solely by the sign of $\ell_t$, regardless of the realization of $s_t$. Thus, when $\ell_t > c_\alpha$, agent $t$ always chooses $a_t = +1$, and when $\ell_t < -c_\alpha$, she always chooses $a_t =-1 $. 

Similarly, when $\lvert \ell_t \rvert = c_\alpha$, the tie-breaking rule at indifference implies that agent $t$ chooses the same action as agent $t-1$, regardless of the realization of $s_t$. Moreover, one can verify that in this case $a_{t-1}$ equals the sign of $\ell_t$.\footnote{To see this, suppose without loss of generality that $\ell_t = c_\alpha$, so $\text{sign}(\ell_t)= +1$. If $s_t =+1$, then $L_t = \ell_t + c_\alpha > c_\alpha$, which leads to $a_t = +1$. If instead $s_t = -1$, then $L_t = \ell_t - c_\alpha =0$ and by the tie-breaking rule, $a_t = a_{t-1}$. In either case, $a_{t-1} =  +1$ because $\ell_t = c_{\alpha}$ implies that either $\ell_{t-1}> c_\alpha$, in which case agent $t-1$ is in the cascade region and chooses $+1$, or $\ell_{t-1}<c_\alpha$ but $a_{t-1}$ must be $ +1$ since $\ell_t = c_\alpha$.} We therefore refer to the region where $\lvert \ell_t \rvert \geq c_\alpha$ as the \emph{cascade region}, in which actions are fully determined by social information. Conversely, when $\lvert \ell_t \rvert < c_\alpha$, the private signal can overturn the public belief, and thus agent $t$ follows her private signal, i.e., $a_t = s_t$. We refer to this region as the \textit{learning region}. 

\section{State and Signal Volatility} \label{sec:state_signal}
To study the long-run volatility of actions, it is useful to first examine the volatility of the underlying state and signals. Since $\theta_t$ follows a simple two-state Markov chain with a symmetric transition probability $\eps$, the expected time between state changes is inversely proportional to the likelihood of a state change. To illustrate this, let $x$ represent the expected time between state changes. Then $x$ satisfies the following equation:
\[
x = \eps + (1-\eps) (1+x),
\]
which implies that $x = 1/\eps$. That is, a higher likelihood of a state change corresponds to a shorter average time between changes. We denote the fraction of periods $t \leq n$ in which the state changes by
\[
\mathcal{Q}_\theta(n) := \frac{1}{n}\sum_{t=1}^n \mathbbm{1}(\theta_t \neq \theta_{t+1}).
\]
By the strong law of large numbers, 
\[
\lim_{n \to \infty}  \mathcal{Q}_{\theta}(n) = \eps \quad \text{almost surely.}
\]
Equivalently, an average state duration\textemdash defined as the number of consecutive periods during which the state remains unchanged\textemdash of $1/\eps$ implies that the state changes in a fraction $\eps$ of periods in the long run. Similarly, we denote the fraction of periods $t \leq n$ in which the signal changes by
\[
\mathcal{Q}_s(n) := \frac{1}{n}\sum_{t=1}^n \mathbbm{1}(s_t \neq s_{t+1}).
\]
Since each signal matches the current state with probability $\alpha$, it also follows from the strong law of large numbers that
\[
\lim_{n\to \infty} \mathcal{Q}_s(n) = \Prob[s_t \neq s_{t+1}] \quad \text{almost surely,}
\]
and a direct calculation yields
\begin{equation} \label{eq:signal_switches}
\Prob[s_t \neq s_{t+1}]= \eps+2(1-2\eps)\alpha(1-\alpha).
\end{equation}

As is evident from \eqref{eq:signal_switches}, this probability is strictly greater than $\eps$, so signals switch more often than the state in the long run. This additional volatility arises simply because signals are not perfectly informative about the state. As signals become more precise, i.e. $\alpha \to 1$, this excess volatility vanishes. In this limit, signals switch as often as the underlying state.   

\section{The Main Result}
We are now ready to state our main result.  
Recall that in our model, each agent receives a private signal about the current state and observes the history of past actions. As in Section \ref{sec:state_signal}, we measure action volatility by the fraction of periods $t \leq n$ in which the action changes:
\begin{equation}\label{eq:action_changes}
\mathcal{Q}_a(n) := \frac{1}{n}\sum_{t=1}^n \mathbbm{1}(a_t \neq a_{t+1}).    
\end{equation}

\begin{theorem} \label{cor:exptime}
For any $\alpha \in (1/2, 1)$ and $\eps \in (0, \alpha(1-\alpha))$,
\[
\lim_{n\to \infty} \mathcal{Q}_s(n)  > \limsup_{n\to \infty} \mathcal{Q}_a(n)  \ge  \liminf_{n\to \infty} \mathcal{Q}_a(n) > \lim_{n\to \infty} \mathcal{Q}_\theta(n)  \quad \text{almost surely.}
\]
\end{theorem}
This result shows that in the long run, actions are less volatile than signals but more volatile than the state. Since the state is relatively stable,  temporary information cascades arise, during which agents stop responding to their private signals and instead mimic their predecessors. These cascades naturally dampen action volatility relative to signals. Perhaps surprisingly, Theorem \ref{cor:exptime} shows that, despite the inertia generated by cascades, actions remain more volatile than the state in the long run. 


For example, consider a regime in which the state changes with probability $1\%$ per period and private signals match the current state with probability $80\%$. In this case, the state changes once every 100 periods on average, whereas \eqref{eq:signal_switches} implies that signals change roughly once every three periods. 
Meanwhile, as we will see later, the average duration between action changes is strictly greater than three periods and bounded above by 44 periods. Hence, in the long run, action switches occur more frequently than state switches but less frequently than signal switches.


Note that because action changes are not independent events, we cannot directly relate the expected time between changes to the corresponding long-run frequency using the standard law of large numbers. Nevertheless, each agent's action is a function of the public belief, which itself is a Markov process.\footnote{Since the state changes over time, the agent's public belief ceases to be a martingale, which is an important tool in analyzing long-term learning outcomes in fixed-state models
\citep{smith2000pathological}.} More specifically, for any $t\ge1$, $a_t = \text{sign}(q_{t+1} - 1/2)$. We use this Markov property to study how the public belief transitions across different regions. These transitions then allow us to derive bounds on the expected time between action switches, which are crucial in proving our main result.

Formally, the proof of Theorem \ref{cor:exptime} proceeds in three steps. First, we show that the expected time between consecutive entries of the public belief process into the cascade region is uniformly bounded (Proposition \ref{prop:Jbound}). This implies that cascade entries\textemdash during which signals may change while actions do not\textemdash  occur with positive frequency in the long run, thereby reducing action switches relative to signal switches. Second, we show that when the belief is in the learning region, the conditional probability of an action switch is uniformly bounded away from zero. Combining this with the bound on the cascade length, we obtain an upper bound on the expected time between action switches, which is strictly less than the expected time between state changes 
(Proposition \ref{prop:ub}). Third, we use a martingale argument to translate these bounds into the almost sure comparisons in Theorem \ref{cor:exptime}, establishing that actions are less volatile than signals but more volatile than the state in the long run.


\section{Analysis}
In this section, we analyze how the agent’s public belief evolves in both the learning and cascade regions. Since each agent's action corresponds to whether the next period's public belief lies above or below $1/2$, these dynamics allow us to establish bounds on the expected time between action changes, which will be useful in proving our main result. We provide a proof sketch of Theorem \ref{cor:exptime} at the end of the section.

\subsection{Belief Dynamics}
When the state is changing, belief dynamics differ between the learning and cascade regions. 
We now describe these dynamics. 

Consider first the case where the public LLR at time $t$ is in the cascade region, i.e., $|\ell_t| \geq c_\alpha$. Suppose $t$ is the time at which the public LLR first enters the cascade region from the learning region. In this case, agent $t$ follows the action of her immediate predecessor, so $a_{t}$ contains no additional information about $\theta_t$ beyond what the history $h_{t-1}$ provides. Meanwhile, between time $t$ and $t+1$, the state may change with probability $\eps$. Since the state follows a Markov chain, conditional on $\theta_t$, the history $h_{t}$ provides no further information about $\theta_{t+1}$. Thus, during a cascade, the corresponding public belief updates deterministically as follows:
\begin{align*}
 q_{t+1}  = \Prob[\theta_{t+1}= +1 \mid h_{t}] = (1-\eps) q_t + \eps (1-q_t) =\ (1-2\eps)q_t +  (2\eps)\frac{1}{2}. 
\end{align*} 

The above expression shows that $q_{t+1}$ tends toward 1/2 while the cascade lasts. This captures the idea that state changes depreciate the value of older social information, since actions observed in earlier periods become less relevant to the current agent. Consequently, the public belief exits the cascade region after finitely many periods. Equivalently, we can write the update in terms of the public LLR as 
\begin{equation} \label{eq:cascade_LLR}
     \ell_{t+1} = G_\eps(\ell_t),
\end{equation}
where 
\[
G_\eps(x) := \log \frac{(1-\eps)\ee^{x} + \eps}{1-\eps + \eps \ee^{x}}.
\]



Next, we consider the learning region where $|\ell_t| < c_\alpha$. Here, $a_t = s_t$, so observing the action reveals the private signal. 
Since the state changes with probability $\eps >0$ in every period, upon observing the latest history, each agent needs to consider the possibility that the state may have changed after the most recent action was taken. 
By Bayes' rule and the law of total probability, the public LLR at time $t+1$ satisfies
\begin{align*} 
\ell_{t+1}  & = \log  \frac{\Prob[\theta_{t+1}= +1 \mid h_{t-1}, a_t]}{\Prob[\theta_{t+1}=-1\mid h_{t-1}, a_t]}  \nonumber \\
& = \log \frac{\sum_{i \in \{-1, +1\}} \Prob[\theta_{t+1}= +1, a_t\mid h_{t-1}, \theta_t=i] \cdot \Prob[\theta_t=i \mid h_{t-1}]}{\sum_{i \in \{-1, +1\}} \Prob[\theta_{t+1}=-1, a_t\mid h_{t-1}, \theta_t =i] \cdot \Prob[\theta_t=i \mid h_{t-1}]}. 
\end{align*}

Moreover, since $\theta_t$ follows a Markov chain, conditional on $\theta_t$, both $\theta_{t+1}$ and $a_t$ are independent of the history $h_{t-1}$ and of each other. Therefore, if $s_t= +1$, the public LLR evolves according to
\[
\ell_{t+1} = \log \frac{(1-\eps) \alpha \ee^{\ell_t} + \eps(1-\alpha)}{\eps \alpha \ee^{\ell_t} + (1-\eps)(1-\alpha)}.
\]
If $s_t= -1$, it evolves as 
\[
\ell_{t+1} = \log \frac{(1-\eps) (1-\alpha) \ee^{\ell_t} + \eps \alpha}{\eps (1-\alpha) \ee^{\ell_t} + (1-\eps)\alpha}.
\]

In sum, when $\ell_t$ is in the learning region, its evolution can be written as 
\begin{equation} \label{eq:learning_LLR}
    \ell_{t+1} =  G_\eps(\ell_t + c_\alpha \cdot s_t).
\end{equation}

Note that $G_\eps(\cdot)$ is strictly increasing. Thus, in the learning region, belief updating is monotone in both the public belief and the realized signal. Under the assumption that $\eps < \alpha(1-\alpha)$, we also have $G_\eps(\ell + c_\alpha) > \ell > G_\eps(\ell - c_\alpha)$ in the learning region. Moreover, because the state is evolving, belief updating is attenuated toward the uniform belief while preserving the sign of its input: for all $x>0$, $0<G_\eps(x) < x$ and for all $x<0$, $0>G_\eps(x) > x$.

Given \eqref{eq:cascade_LLR} and \eqref{eq:learning_LLR}, the belief dynamics in the standard fixed-state model can now be obtained as a special case by setting $\eps = 0$. In this case, $G_0(x)=x$, so the public LLR evolves according to $\ell_{t+1} = \ell_t + c_\alpha \cdot s_t$ if $|\ell_t| < c_\alpha$ and $\ell_{t+1} = \ell_t$ if $|\ell_t| \ge c_\alpha$. Thus, in the learning region, the next agent’s public belief coincides with the current agent’s posterior belief, so \(\ell_{t+1}=L_t\). In the cascade region, the public belief remains unchanged, since all subsequent agents face the same decision problem as the agent who initiated the cascade. Consequently, once a cascade forms, it persists indefinitely.

 Another important observation is that, whether the state is fixed or changing, the process $(\ell_t)$ forms a Markov chain.\footnote{This is because conditional on the state $\theta_t$, the private signal $s_t$ is independent of $\ell_{\tau}$, for any $\tau < t$.} When the state is fixed, the state space of this Markov chain is finite since the magnitude difference between $\ell_t$ and $\ell_{t+1}$ is a constant for any given signal precision. However, when the state is changing, the state space becomes infinite, as these magnitude differences also depend on the current value of $\ell_t$. This complicates the calculation of both 
 the exact expected time between sign switches and
 the stationary distribution of beliefs. We circumvent this challenge by providing bounds on this expected time instead.



\subsection{Bounds on Expected Durations} 
We first establish two properties of the public LLR dynamics in the learning region. At any time $t$ with $\ell_t \neq 0$, we say that an action is \textit{opposing} to the current public belief if $a_t \neq \text{sign}(\ell_t)$ and \textit{supporting} otherwise. 

\begin{lemma} \label{lem:steps}

Suppose the public LLR at time $t$ satisfies $0<|\ell_t| < c_\alpha$. Then the following two properties hold.
\begin{compactenum}[(i)]
\item If $a_{t} \neq \text{sign}(\ell_t)$, then $\text{sign}(\ell_{t+1}) =-\text{sign}(\ell_{t})$. 
\item There exists a finite integer $N = N(\alpha, \eps)$, independent of $\ell_t$, such that 
\[
a_{t}=\ldots = a_{t+N-1}= \text{sign}(\ell_t) \Rightarrow \lvert \ell_{t+k}\rvert \ge c_\alpha \quad \text{for some~} k \in \{1, \ldots, N\}.
\]
\end{compactenum}
\end{lemma}

The first part of this lemma states that when the public belief is in the learning region, a single opposing action is sufficient to overturn the sign of the public LLR. This is close in spirit to the \textit{overturning principle}  \citep{smith2000pathological}, but applies here in a setting with a changing state. The second part shows that initiating a cascade requires at most $N$ consecutive supporting actions, where $N$ depends only on $\alpha$ and $\eps$. For example, when $\alpha =0.8$, a cascade requires at most two supporting actions if $\eps = 0.01$, and at most five if $\eps = 0.159$. Intuitively, when the state is less stable, social information depreciates faster, so triggering a cascade may require a longer sequence of supporting actions. 

Next, we consider the maximum length of any cascade. During a cascade, the public LLR drifts deterministically toward zero until it re-enters the learning region. By \eqref{eq:learning_LLR}, its magnitude at cascade entry is bounded above by $G_\eps(2c_\alpha
)$, and so no cascade can last longer than one starting at this upper bound. Using \eqref{eq:cascade_LLR}, we can therefore calculate an upper bound on the length of any cascade, which we denote by
\[
K(\alpha, \eps) := \frac{\log (1-2\alpha(1-\alpha))}{\log (1-2\eps)}.
\]
For completeness, we provide a derivation of $K(\alpha, \eps)$ in the appendix (see Lemma \ref{lem:cascade_length}), using a calculation analogous to that in Section 3.B of \citet*{moscarini1998social}. Note that $K(\alpha, \eps)$ decreases in both $\alpha$ and $\eps$. Intuitively, as private signals become less precise, less social information is needed to outweigh them, potentially extending the duration of a cascade. Similarly, as the state becomes less volatile, cascades can last longer because social information depreciates more slowly.

Building on these dynamics and the bound on cascade length, we now establish two upper bounds. The first concerns the expected time between consecutive entries into the cascade region. Since actions do not change during cascades while private signals continue to fluctuate, this bound is useful for comparing the long-run frequency of action switches with that of signal switches. 

Formally, for each integer $i = 1, 2, \ldots$, let $\mathcal{C}_i$ denote the (random) time at which the public LLR enters the cascade region from the learning region for the $i$-th time, with the convention $\mathcal{C}_0 = 0$. Define  $\mathcal{J}_i := \mathcal{C}_i - \mathcal{C}_{i-1}$ as the duration between the $(i-1)$-th and $i$-th cascade entries.   

\begin{proposition} \label{prop:Jbound}
There exists a finite integer $N$ such that,  for any integer $i \ge 2$, conditional on the public LLR entering the cascade region from the learning region for the $(i-1)$-th time, the expected time until the next entry into the cascade region satisfies
\[
\mathbb{E}[\mathcal{J}_i \mid \ell_{\mathcal{C}_{i-1}}] \le \lfloor K(\alpha,\eps)\rfloor +  2^N \cdot N.
\]
\end{proposition}
The proof of Proposition \ref{prop:Jbound} relies on the observation that the public belief process alternates between the cascade and learning regions. Accordingly, the bound consists of two components: the maximum duration of a cascade and the expected time spent in the learning region. By Lemma \ref{lem:steps},  $N$ consecutive supporting signals suffice to trigger a cascade. Since each supporting signal has conditional probability at least 1/2, the probability of entering a cascade within the next $N$ periods, conditional on being in the learning region, is at least $2^{-N}$. Thus, the expected time spent in the learning region is uniformly bounded by $2^N N$. 

Our second upper bound concerns the expected time between sign switches of the public LLR. As mentioned before, since $a_t = \text{sign}(\ell_{t+1})$ for all $t \ge 1$, this also bounds the expected time between action switches. Formally, for each integer $i = 1, 2, \ldots$, let $\mathcal{T}_i$ denote the (random) time at which the public LLR switches sign for the $i$-th time, with the convention $\mathcal{T}_0 = 0$. Define $\mathcal{D}_i := \mathcal{T}_i - \mathcal{T}_{i-1}$
as the duration between the $(i-1)$-th and $i$-th sign switches. 

\begin{proposition} \label{prop:ub}
For any integer $i \geq 2$, conditional on the public LLR switching its sign for the $(i-1)$-th time, the expected time until the next sign switch satisfies
\begin{align*}
    \E[\mathcal{D}_i \mid \ell_{\mathcal{T}_{i-1}}] <  M(\alpha, \eps) < \frac{1}{\eps}
\end{align*}
where 
\[
M(\alpha, \eps) := \frac{1+(1-2\alpha(1-\alpha))K(\alpha, \eps)}{2\alpha(1-\alpha)}. 
\]
\end{proposition}
This result implies that, on average, the public LLR  changes its sign at least once every $M(\alpha, \eps)$ periods, an interval strictly shorter than the average duration between state changes. For example, when $\alpha = 0.8$ and $\eps =0.01$, $M(0.8, 0.01)$ is approximately $44$. Thus, the expected time between action switches is less than 44 periods, whereas the expected time between state switches is 100 periods. 

The proof idea behind Proposition \ref{prop:ub} is as follows. In the learning region, a supporting action has conditional probability strictly below $1-2\alpha(1-\alpha)$, and since agents follow their private signals in this region, a sign switch has conditional probability strictly above $2\alpha(1-\alpha)$. Thus, the number of supporting actions in the learning region before a sign switch is stochastically dominated by a geometric random variable with success probability \(2\alpha(1-\alpha)\). Note that each run between sign switches starts in the learning region, and each supporting action, together with any cascade it triggers, takes at most $K(\alpha, \eps) +1$ periods. Allowing one more period for the final opposing action, the expected time until the next sign switch is strictly less than: 
\[
1 +2\alpha(1-\alpha) \sum_{i=1}^\infty i \cdot (1-2\alpha(1-\alpha))^i \big[ K(\alpha, \eps) +1\big] = M(\alpha, \eps). 
\]
As we show in the appendix, the remaining inequality $M(\alpha, \eps) < 1/\eps$ holds for all \(\alpha\in(1/2,1)\) and \(0<\eps<\alpha(1-\alpha)\).

\subsection{Proof Sketch of Theorem \ref{cor:exptime}} We end this section by providing a proof sketch of Theorem \ref{cor:exptime}. First, to see that action volatility is strictly lower than signal volatility, observe that upon each entry into a cascade, there is a uniformly positive probability that signals switch while actions remain constant. Since the expected time between successive entries is uniformly bounded above (Proposition \ref{prop:Jbound}), such events occur with positive frequency in the long run.
Second, to see why action volatility is strictly higher than state volatility, Proposition \ref{prop:ub} bounds the expected time between action switches from above by $M(\alpha, \eps)< 1/\eps$, where \(1/\eps\) is the expected time between state changes. In the appendix, we use a strong law for martingales to convert these bounds on expected durations into bounds on long-run frequencies. Together, these arguments establish that actions switch less often than signals but more often than the state in the long run.

\section{When  Signals Become Public} \label{sec:public_signal}
So far, we have focused on the public-action setting in which each agent observes her own private signal and the actions of her predecessors. While information cascades arise temporarily, Theorem \ref{cor:exptime} shows that they are insufficient to prevent actions from switching more frequently than the state. We now consider a public-signal setting in which each agent $t$ not only observes her own signal but also those of all her predecessors. That is, before choosing her action, agent $t$ observes a collection of private signals $(s_1, \ldots, s_t)$, each matching the corresponding state with probability $\alpha$. We study how often agents' actions change in this setting and compare the long-run action volatility here with that in the public-action setting.

A priori, it is unclear which setting generates greater action volatility. On the one hand, because past actions are coarse summaries of agents' private signals, observing signals directly provides more precise information about the state, which could reduce unnecessary action switches. On the other hand, when all signals are publicly observed, information cascades no longer arise, so agents continue to incorporate each new signal into their beliefs and decisions. Compared with cascade periods, during which agents stop responding to their signals and mimic their predecessors, this continued responsiveness could thus lead to more frequent action switching. As we show in the next two propositions, action volatility can, in fact, be higher or lower in either setting, depending on which force dominates. 

To formalize this comparison, we use $a_t^S$ to denote the action taken by agent $t$ in the public-signal setting. Analogously to \eqref{eq:action_changes}, we use $\mathcal{Q}_a^S(n) := \frac{1}{n}\sum_{t=1}^n \mathbbm{1}(a^S_t \neq a^S_{t+1})$ to denote the fraction of periods $t\leq n$ in which the action changes in this setting. The first proposition identifies a regime in which agents observe more information yet switch actions more frequently.
\begin{proposition} \label{prop:new1}
There exists a sufficiently small $\delta>0$ such that if 
$\alpha<1/2+\delta$ and $\eps \leq (2\alpha-1)^2$, then
\[
\liminf_{n\to\infty} \mathcal{Q}_a^S(n)
>
\limsup_{n\to\infty}\mathcal Q_a(n)
\quad\text{almost surely}.
\]
\end{proposition}
Together with Theorem \ref{cor:exptime}, this result shows that action volatility can still exceed state volatility even in the absence of information cascades. This suggests that excess action volatility does not arise solely from the loss of information during the social learning process. Even with access to all past signals, agents may still switch actions more often than the underlying state.

The intuition behind Proposition \ref{prop:new1} is as follows. When signals are imprecise and the state is very stable, past information depreciates slowly. This means that cascades last longer, which reduces action volatility in the public-action setting. Meanwhile, agents who observe all past signals continue to update their beliefs, so their actions remain responsive to these noisy signals. The restriction $\eps \le (2\alpha-1)^2$ further ensures that the state changes slowly enough so that information from weak signals can accumulate before it becomes outdated. In this regime, the stabilizing effect of information cascades outweighs the advantage of observing more information, leading to higher action volatility in the public-signal setting.

The next proposition identifies a regime in which the opposite conclusion holds. 

\begin{proposition} \label{prop:new2}
   Fix $\gamma \in (1/2, 1)$. There exists a sufficiently small $\delta>0$ such that if $\alpha > 1-\delta$ and $\eps = \gamma(1-\alpha)$, then
    \[
\limsup_{n\to\infty} \mathcal{Q}_a^S(n)
< 
\liminf_{n\to\infty}\mathcal Q_a(n)
\quad\text{almost surely}.
\]
\end{proposition}
This proposition shows that when signals are very precise and state changes are sufficiently frequent relative to the arrival of misleading signals, agents switch their actions more often when they observe all past actions rather than all past signals. 

Intuitively, under these conditions, cascades are short, and so agents in the public-action setting frequently resume following their private signals. An occasional misleading signal can therefore trigger an action switch even when the state remains unchanged. In contrast, in the public-signal setting agents can use the accumulated information to offset an isolated opposing signal and keep their action unchanged. In this regime, the advantage of observing more information about a changing state dominates the stabilizing effect of information cascades, leading to lower action volatility in the public-signal setting.

\subsection*{Numerical Simulations}
 \begin{figure}[t!]
        \centering
        
        \includegraphics[width=0.7\linewidth]{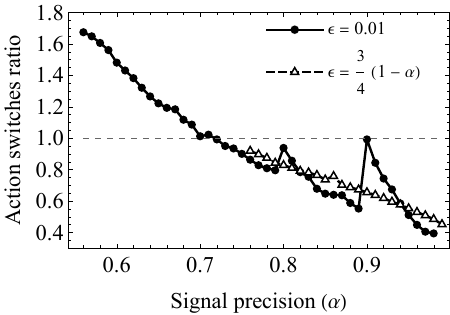}
        \caption{The ratio of action switch frequencies when past signals are public to those when only past actions are public as a function of signal precision $\alpha$ for $\eps = 0.01$ and $\eps = \frac{3}{4}(1-\alpha)$. The dashed horizontal line corresponds to a ratio of one.}
         \label{fig:ratio_public}
    \end{figure}

To further illustrate the results, we use Monte Carlo simulations to compare action volatility across the two settings. We vary signal precision under two specifications: $\eps =0.01$ and $\eps = \frac{3}{4} (1-\alpha)$. For each parameter pair $(\alpha, \eps)$, we simulate 100,000 periods in each setting, starting from a uniform prior and using the same realized sequence of states. We compute the fraction of periods in which the action switches in each setting and divide the frequency in the public-signal setting by that in the public-action setting. Figure \ref{fig:ratio_public} plots this ratio as a function of signal precision. A ratio above one indicates that action volatility is higher in the public-signal setting; a ratio below one indicates the opposite. 

Consistent with Propositions \ref{prop:new1} and \ref{prop:new2}, the figure shows that observing signals rather than actions sometimes reduces long-run action volatility and sometimes increases it. For example, when signals are imprecise ($\alpha = 0.6$) and the state is sufficiently stable relative to signal precision ($\eps = 0.01 < (2\alpha-1)^2 = 0.04$), the ratio is about 1.5. Thus, agents in the public-signal setting switch actions about 50\% more frequently than those in the public-action setting, despite having more information. In contrast, when the state changes at three-quarters the rate of misleading signals ($\eps = 3(1-\alpha)/4$), we see that the ratio is below one for sufficiently precise signals. In this regime, agents switch actions more frequently in the public-action setting than in the public-signal setting, despite the presence of information cascades.

\section{Conclusion}
In this paper, we study the behavior of agents who receive a private signal and observe the past actions of their predecessors in a changing environment. As the state evolves, the best action to take also fluctuates. We show that, in the long run, actions switch less frequently than signals\textemdash due to the presence of information cascades that arise temporarily\textemdash but still more frequently than the state itself. 


One may wonder whether the main result is driven by the high frequency of action changes when the posterior belief is around 1/2. Accordingly, we can further restrict the measurement of action volatility to exclude consecutive action changes, i.e., cases where $a_t \neq a_{t-1}$ and $a_{t-1}\neq a_{t-2}$. Simulation results show that even under this more restrictive counting, actions still change more frequently than the state. For example, with $\alpha = 0.75$, $\eps = 0.05$, and a total of $100, 000$ periods, action changes occur approximately 8,150 times, compared to the number of state changes, which is about 5,100 times.

There are several possible avenues for future research. As seen in Section \ref{sec:public_signal}, action volatility can be higher or lower in the public-signal setting than in the public-action setting. It would therefore be interesting to fully characterize the regions in which each information structure generates higher action volatility. It is also natural to study the corresponding long-run mistake probabilities under these information regimes. We leave these interesting questions for future work, as they currently seem to be beyond what is technically tractable.

Another open question concerns the tightness of the bound in Proposition \ref{prop:ub}. In particular, one could ask whether this upper bound $M(\alpha, \eps)$ is tight, and if so, whether for a large finite time $N$, the number of action changes is well approximated by $N/M(\alpha, \eps)$. Our simulation results suggest that this is not the case. For example, when $\alpha = 0.9$ and $\eps =0.05$, and $N= 100,000$, we have $M(0.9, 0.05) \approx 14.1$, which would predict at least about 7,100 action changes. In contrast, our simulations exhibit approximately 14,200 action changes\textemdash about twice as many as implied by the bound. Furthermore, the simulations suggest that as $\alpha-1/2$ and $\eps$ approach zero at the same rate, the ratio between the frequencies of action changes and state changes approaches a constant that is close to 4.  This suggests that achieving a very accurate understanding of action volatility in this regime might be possible.

\appendix
\section{Omitted Proofs}

\begin{proof}[Proof of Lemma \ref{lem:steps}] 
Consider the case where $\ell_t \in (0, c_\alpha)$. For part (i), an opposing action implies that $s_t =a_t =-1$. Since $G_\eps$ is strictly increasing and $G_\eps(0) =0$, \[
\ell_{t+1} = G_\eps(\ell_t - c_\alpha)<0.
\]
Thus, $\text{sign}(\ell_{t+1}) = -1 =  -\text{sign}(\ell_t)$.

For part (ii), define a deterministic sequence $(x_n)$ by $x_0 = 0$ and $x_{n+1} = G_\eps(x_n + c_\alpha)$ for all $n\geq 0$. Denote the number of positive updates required for this sequence to enter the cascade region by  
\[
N(\alpha, \eps) := \min \{n \geq 1: x_n \ge c_\alpha\}.
\]
We first show that $N(\alpha, \eps)$ is finite. Since $\eps < \alpha (1-\alpha)$, a direct calculation shows that $G_\eps(2c_\alpha) > c_\alpha$. Moreover, $0 < G'_\eps(x)\leq 1-2\eps <1 $, so the function $G_\eps(x+c_\alpha) - x$ is decreasing in $x$. Thus, for every $x \in [0, c_\alpha]$, 
\[
G_\eps(x+c_\alpha) - x \geq G_\eps(2c_\alpha)  - c_\alpha >0.
\]
Hence, each increment before reaching $c_\alpha$ is bounded below by a positive constant, so  $N(\alpha, \eps)< \infty$. 

For ease of notation, write $N = N(\alpha, \eps)$. Suppose now that $\ell_t \in (0,c_\alpha)$ and
$a_t = \cdots = a_{t+N-1} = +1$. If a cascade is entered within the first
$N-1$ updates, the conclusion follows. Otherwise, the public LLR remains in the learning region at times $t, \ldots, t+N-1$, so all $N$ actions reveal positive signals. Applying the strict monotonicity of $G_\eps$ inductively gives
\[
\ell_{t+j} >x_j, \quad j= 0, \ldots, N.
\]
In particular, $\ell_{t+N} > x_N \geq c_\alpha$. Hence, a cascade is entered
within $N$ updates. The case $\ell_t \in (-c_\alpha,0)$ follows by symmetry. 
\end{proof}

\begin{lemma} \label{lem:cascade_length}
    The length of any information cascade is bounded above by   $$K(\alpha, \eps)= \frac{\log (1-2\alpha(1-\alpha))}{\log (1-2\eps)}.$$
\end{lemma}

\begin{proof}
    Without loss of generality, we consider a positive cascade. Since the cascade is entered from the learning region, its initial public LLR is strictly less than  $G_{\eps}(2c_\alpha)$. Its corresponding public belief is thus strictly less than 
    \[
    m := \frac{(1-\eps)\alpha^2 + \eps (1-\alpha)^2}{\alpha^2+(1-\alpha)^2}.
    \]
    As seen in the main text, within the cascade region, the public belief evolves deterministically towards 1/2. Starting from an initial value $x$, iterating this update rule $h$ times yields
    \begin{align*}
        g(h, x) :=  & \eps \sum_{i=0}^{h-1}(1-2\eps)^i + (1-2\eps)^h x\\
         = & \frac{1}{2} + (1-2\eps)^h \left(x - \frac{1}{2} \right).
    \end{align*}
    Consequently, if the cascade lasts $h$ periods, the public belief after these updates is strictly less than $g(h, m)$. The cascade terminates when the belief drops below $\alpha$, so it suffices to find $h$ such that $g(h, m) \leq \alpha$. Since 
     \[
        m - \frac{1}{2} = \frac{(1-\eps)\alpha^2 + \eps (1-\alpha)^2}{\alpha^2+(1-\alpha)^2} - \frac{1}{2} = \frac{(1-2\eps)(2\alpha - 1)}{2(\alpha^2+(1-\alpha)^2)},
    \]
    this condition is equivalent to 
    \[
    (1-2\eps)^{h+1} \leq 1-2\alpha(1-\alpha).
    \]
   Taking $h=\lfloor K(\alpha,\eps)\rfloor$ satisfies this condition, so every cascade lasts at most \(\lfloor K(\alpha,\eps)\rfloor\le K(\alpha,\eps)\) periods. 
\end{proof}

We introduce the following notation, which will be useful in the subsequent proofs. Let $\pi(\ell)$ denote the probability of receiving a positive signal conditional on the public LLR being $\ell$, suppressing its dependence on $\alpha$. By the law of total probability, 
\begin{align*}\label{eq:prob_upsignal}
    \pi(\ell)  = \alpha \cdot \frac{\ee^\ell}{1+\ee^\ell} + (1-\alpha)\cdot \frac{1}{1+\ee^\ell} = \frac{1+ \alpha(\ee^\ell -1)}{1+ \ee^\ell},
\end{align*}
which is strictly increasing in $\ell$.  Therefore, 
\begin{equation}\label{eq:max_prob}
    \sup_{\ell \in (0, c_\alpha)}\pi(\ell) = 1- 2\alpha(1-\alpha) =: \bar \pi.
\end{equation} 
Note that $\pi(0) = 1/2$ and $\pi(-\ell) = 1-\pi(\ell)$ by symmetry. Consequently, for any $\ell \neq 0$, the probability of receiving a supporting signal is $\pi(|\ell|) > 1/2$.

\begin{proof}[Proof of Proposition \ref{prop:Jbound}]
We decompose $\mathcal{J}_i$ into the time spent in the cascade region and the subsequent time spent in the learning region. By Lemma \ref{lem:cascade_length}, the former is bounded by $\lfloor K(\alpha,\eps)\rfloor$.

For the latter, let $N = N(\alpha, \eps)$ be the integer from part (ii) of Lemma \ref{lem:steps}. Then $N$ consecutive supporting actions ensure re-entry into the cascade region within $N$ updates. Along this sequence, the public LLR has the same sign, and a supporting action has conditional probability at least $1/2$ in the learning region and probability one in the cascade region. Thus, for any public history $h_{t-1}$ with $0<|\ell_{t}| < c_\alpha$, applying the chain rule of conditional probability repeatedly gives 
\begin{align*}
      \Prob[a_t = \ldots = a_{t+N-1} = \text{sign}(\ell_t) \mid h_{t-1}]  
     \geq  2^{-N}.   
\end{align*}

Consider consecutive blocks of \(N\) updates starting from entry into the learning region. If cascade entry has not occurred within a block, the same conditional bound applies to the next block. The probability of no cascade entry in the first \(k\) blocks is thus at most \((1-2^{-N})^k\). Hence, the waiting time from entering the learning region to the next cascade entry is stochastically
dominated by $N G$, where $G$ is a geometric random variable on $\{1, 2, \ldots \}$ with success probability $2^{-N}$. Combining the two bounds gives
\[
\E[\mathcal J_i\mid\ell_{\mathcal C_{i-1}}]
\le\lfloor K(\alpha,\eps)\rfloor+N\E[G]
=\lfloor K(\alpha,\eps)\rfloor+N2^N.
\]
\end{proof}

\begin{proof}[Proof of Proposition \ref{prop:ub}]
Fix $i \geq 2$. Recall that $\mathcal{D}_i = \mathcal{T}_{i} -\mathcal{T}_{i-1}$. By symmetry, suppose $\ell_{\mathcal{T}_{i-1}} >0$. Since $G_\eps$ preserves the sign, a sign switch can occur only following an opposing action in the learning region. Thus, by \eqref{eq:learning_LLR} and the monotonicity of $G_\eps$, 
\[
0 < \ell_{\mathcal{T}_{i-1}} = G_\eps(\ell_{\mathcal{T}_{i-1}-1}+c_\alpha) \le G_\eps(c_\alpha) < c_\alpha.
\]
Hence, the run starts in the learning region. 

Let $K := \lfloor K(\alpha, \eps) \rfloor \geq 1$, and let $R_i $ denote the number of learning region updates during the \(i\)-th run, up to and including the final opposing action:
\[
R_i
:=
\sum_{t=\mathcal T_{i-1}}^{\mathcal T_i-1}
\mathbbm{1}(|\ell_t|<c_\alpha).
\]
Each of the first \(R_i-1\) learning region actions can trigger at most one cascade and the final opposing action ends the run. Since each cascade lasts at most \(K\) periods,
\begin{equation} \label{eq:bound_D}
\mathcal D_i\le R_i+K(R_i-1).
\end{equation}
By part (i) of Lemma \ref{lem:steps}, the event $\{R_i >k\}$ requires the first $k$  learning region actions to be supporting. Hence, by \eqref{eq:max_prob}, for every integer $k\ge1$,
\[
\Prob[R_i > k \mid \ell_{\mathcal{T}_{i-1}}] \leq \pi(\ell_{\mathcal{T}_{i-1}}) \cdot \bar \pi^{k-1}
\] 
Consequently,
\begin{align*}
\mathbb E[R_i\mid\ell_{\mathcal T_{i-1}}]
&=
1+\sum_{k=1}^{\infty}
\mathbb P[R_i>k\mid\ell_{\mathcal T_{i-1}}] \leq 
1+\frac{\pi(\ell_{\mathcal T_{i-1}})}{1-\bar\pi}
<
\frac{1}{1-\bar\pi},
\end{align*}
where the last inequality follows from $\pi(\ell) < \bar \pi$ for $\ell \in (0, c_\alpha)$.  

Taking conditional expectations in \eqref{eq:bound_D}, we obtain
\begin{align*}
\mathbb E[\mathcal D_i\mid\ell_{\mathcal T_{i-1}}]
&\le
(K+1)\mathbb E[R_i\mid\ell_{\mathcal T_{i-1}}]-K\\
&<
\frac{K+1}{1-\bar\pi}-K
=
\frac{1+\bar\pi K}{1-\bar\pi} \le
\frac{1+\bar\pi K(\alpha,\eps)}{1-\bar\pi} = M(\alpha, \eps),
\end{align*}
where the last inequality uses \(K\le K(\alpha,\eps)\). Finally, the inequality $M(\alpha, \eps) < 1/\eps$ follows from Claim \ref{claim:strict} below. 
\end{proof}

\begin{claim}\label{claim:strict}
For all $\alpha \in (1/2, 1)$ and $\eps \in (0, \alpha(1-\alpha))$, we have 
\[
    M(\alpha, \eps) < 1/\eps.
\]
\end{claim}
\begin{proof}
Recall that 
\[
\bar\pi=1-2\alpha(1-\alpha),
\qquad
K(\alpha,\eps)
=\frac{\log\bar\pi}{\log(1-2\eps)}.
\]
Since \(\eps<\alpha(1-\alpha)\), we have
$1-\bar\pi=2\alpha(1-\alpha)>2\eps$.
Moreover, applying \(\log x<x-1\) to \(x=1/\bar\pi\) and \(x=1-2\eps\), respectively, gives
\[
-\bar\pi\log\bar\pi<1-\bar\pi,
\qquad
-\log(1-2\eps)>2\eps.
\]
It thus follows that
\begin{align*}
M(\alpha,\eps) =
\frac{1+\bar\pi K(\alpha,\eps)}{1-\bar\pi} &=
\frac{1}{1-\bar\pi}
+
\frac{-\bar\pi\log\bar\pi}
{(1-\bar\pi)\bigl[-\log(1-2\eps)\bigr]}\\
&<
\frac{1}{2\eps}
+\frac{1}{2\eps}
=\frac{1}{\eps},
\end{align*}
as claimed. 
\end{proof}

The following lemma will be useful in the proof of Theorem \ref{cor:exptime}. It establishes that the process $(\mathcal{D}_i)$ has well-defined moments. In particular, it implies that there is a finite uniform upper bound to its second moment $\E[\mathcal{D}_i^2]$, which is required to apply the standard martingale convergence theorem. Intuitively, since any cascade must end after $K(\alpha, \eps)$ periods, the probability that $\mathcal{D}_i$ is larger than some finite periods decreases exponentially fast, and so $\mathcal{D}_i$ must have finite moments. 

\begin{lemma}\label{lem:finitevar}
For every $r \in \{1,2,\ldots\}$, there is a constant $c_r$ that depends on $\alpha$ and $\eps$ such that for all $i \ge 2$,  $\E[|\mathcal{D}_i|^r] < c_r$.
I.e., each moment of $\mathcal{D}_i$ is uniformly bounded, independently of $i$.
\end{lemma}
\begin{proof}
Fix $i \geq 2$. Suppose $\ell_{\mathcal{T}_{i-1}} >0$, so that $\mathcal{D}_i = \mathcal{T}_{i} - \mathcal{T}_{i-1}$ is the time elapsed from a positive public LLR to a negative one. As in the proof of Proposition \ref{prop:ub}, with \(K=\lfloor K(\alpha,\eps)\rfloor\), we have 
\[
\mathbb P[R_i>k]\le\bar\pi^k,
\qquad k= 1,\ldots,
\] 
and \eqref{eq:bound_D} implies that $\mathcal D_i\le (K+1)R_i$. Hence, for every integer $n \ge 2$, we have 
\[
 \Prob[\mathcal{D}_i >n ] \le \Prob\left[R_i>\frac{n}{K+1}\right]
\le
\bar\pi^{\lfloor n/(K+1)\rfloor}.
\]
The above tail bound implies that, for any $p >0$, 
\begin{align} \label{eq:prob_n}
\lim_{n \to \infty} n^p \Prob[|\mathcal{D}_i| > n]  = \lim_{n \to \infty} \frac{n^p}{ 1/\Prob[\mathcal{D}_i > n]} \le  \lim_{n \to \infty} \frac{n^p}{(1/\bar{\pi})^{\lfloor n/(K+1)\rfloor}} =0.
\end{align}
For any $r \geq 1$, the $r$-th moment of $|\mathcal{D}_i|$ satisfies
\begin{align*}
    \E[|\mathcal{D}_i|^r] &= \int_{0}^\infty \Prob[|\mathcal{D}_i|^r >t] dt \\
    & = 1+\int_{1}^\infty \Prob[\mathcal{D}_i > y ] r y^{r-1}dy \\
    & = 1+\sum_{n=1}^{\infty} \int_{n}^{n+1} \Prob[\mathcal{D}_i > y ] r y^{r-1}dy  \\
    & \le 1+ \sum_{n=1}^{\infty} \Prob[\mathcal{D}_i > n] r (n+1)^{r-1}, 
\end{align*}
where the second equality follows from a change of variable $y = t^{1/r}$. Since \eqref{eq:prob_n} implies that $\Prob[\mathcal{D}_i >n ] \le C n^{-p}$ for some nonnegative constant $C$ independent of $i$, it follows that for any $p > r $,
\begin{align*}
    \E[|\mathcal{D}_i|^r] & \le 1+ r C \sum_{n=1}^\infty \frac{(n+1)^{r-1}}{n^p} \\
     & \le 1+ r 2^{r-1} C  \sum_{n=1}^\infty  \frac{1}{n^{p-r+1}}< \infty,
\end{align*}
which holds for all $i \ge 2$. Hence, for every $r \in \{1, 2, \ldots\}$, there exists a constant $c_r = 2+ r 2^{r-1} C  \sum_{n=1}^\infty  \frac{1}{n^{p-r+1}}$, independent of $i$, that uniformly bounds $\E[|\mathcal{D}_i|^r]$ for $i \ge 2$. 
\end{proof}

The following lemma will also be useful in the proof of Theorem \ref{cor:exptime}. 
Define the number of wasted signal switches up to time $n$ by 
\[
W_n := \sum_{t=1}^n \mathbbm{1}(s_t \ne s_{t+1}) \cdot \mathbbm{1}(a_t = a_{t+1} = s_{t+1}).
\]
That is, $W_n$ counts times at which the signal switches but the action does not, with the unchanged action matching the new signal. Recall that $n \cdot \mathcal{Q}_a(n)$ and $n \cdot \mathcal{Q}_s(n)$ count the total number of action and signal switches by time $n$, respectively.  

\begin{lemma} \label{lem:action_signal}
For every $n \ge 1$,
\begin{equation} \label{eq:counting}
n \cdot \mathcal{Q}_a(n) \le n \cdot \mathcal{Q}_s(n) - W_n + 1.
\end{equation}
\end{lemma}

\begin{proof}
For any fixed $n \geq 1$, partition $\{1,\dots,n+1\}$ into $n \mathcal{Q}_s(n) + 1$ maximal signal runs. Within each run, the signal $s_t$ is constant, say $s_t = s \in \{-1, +1\}$. 

We first show that within any signal run, the action $a_t$ switches at most once. Suppose $u$ and $u+1$ belong to the same run and $a_u= s$. If the next period is in the learning region, $a_{u+1} = s_{u+1} = s$. If it is in the cascade region, $a_{u+1} = \text{sign} (\ell_{u+1}) = s$, since $\ell_{u+2} = G_\eps(\ell_{u+1})$ and $G_\eps$ preserves the sign. By induction, once $a_t$ aligns with the signal $s$, it remains constant for the remainder of the run. Thus, the only possible action switch within a run is a single transition from $-s$ to $s$.

Next, consider a signal switch between periods $t$ and $t+1$, which starts a new run. Let  $s_{t+1} = s$. Then there are three possible cases: (i) if $a_t \neq a_{t+1}$, then $t+1$ must be in the learning region since an agent in a cascade mimics the action of her predecessor. Hence $a_{t+1} = s$. This contributes one action switch at time $t$, and the action remains unchanged throughout the new run; (ii) if $a_t = a_{t+1} = s$, then by the argument above, the action remains unchanged  throughout the new run. This signal switch is thus counted in $W_n$ and contributes no action switches;
(iii) if $a_t = a_{t+1} = -s$, the action may switch at most once within the new run.

Hence, each of the $n \mathcal{Q}_s(n)$ signal switches and its subsequent run together contribute at most one action switch, except those counted by $W_n$, which contribute none. Finally, the initial signal run (before the first signal switch) can contribute at most one action switch, so we obtain \eqref{eq:counting}.
\end{proof}

Now, we are ready to prove the main theorem. 

\section*{Proof of Theorem \ref{cor:exptime}}
We divide the proof of Theorem \ref{cor:exptime} into two parts. The first part proves the lower bound. The second part proves the upper bound. 

\subsection*{Part (i)}
We first prove that $\liminf_{n\to \infty} \mathcal{Q}_a(n)  > \lim_{n\to \infty} \mathcal{Q}_\theta(n)$ almost surely. Recall that $\mathcal{D}_i = \mathcal{T}_i - \mathcal{T}_{i-1}$ is the (random) duration of the $i$-th run of $\ell$, where a run is defined as the number of consecutive periods during which the public LLR has the same sign. Let $\Phi = (\mathcal{F}_1, \mathcal{F}_2, \ldots)$ be the filtration where each $\mathcal{F}_i = \sigma(\mathcal{D}_1, \ldots, \mathcal{D}_i)$ and thus $\mathcal{F}_{j} \subseteq \mathcal{F}_i$ for any $j \leq i$. The process $(\mathcal{D}_1, \mathcal{D}_2, \ldots)$ is adapted to $\Phi$ since each $\mathcal{D}_i$ is $\mathcal{F}_i$-measurable. By Proposition \ref{prop:ub} and Claim \ref{claim:strict}, there exists $\delta = 1/\eps - M(\alpha, \eps)>0$ such that for all $i \geq 2$,
\[
\E[\mathcal{D}_i \mid \ell_{\mathcal{T}_{i-1}}] < 1/\eps - \delta.
\]
By the law of iterated expectation and the Markov property of the public LLR, 
\begin{equation} \label{eq:prop1}
\E[\mathcal{D}_i \mid \mathcal{F}_{i-1}] = \E\left[\E[\mathcal{D}_i \mid \ell_{\mathcal{T}_{i-1}}, \mathcal{F}_{i-1}] \mid \mathcal{F}_{i-1}\right] < 1/\eps -\delta. 
\end{equation}
Let $$X_{i} := \mathcal{D}_i - \E[\mathcal{D}_i\mid \mathcal{F}_{i-1}] \quad \text{for all}~ i \ge 2.$$ 
Since $\mathcal{F}_{i-1} \subseteq \mathcal{F}_i$, each $X_i$ is $\mathcal{F}_i$-measurable. Define the partial sum process by $S_1 := 0 $, and for $k\ge 2$, 
\[
S_k := X_2 + \frac{1}{2}X_3 + \cdots + \frac{1}{k-1}X_k.
\]
By definition, $\E[X_i \mid \mathcal{F}_{i-1}] =0$  for all $i \geq 2$. Since each $S_{k-1}$ is $\mathcal{F}_{k-1}$-measurable, 
\begin{align*}
    \E[S_k \mid \mathcal{F}_{k-1}] & = \E\left[\sum_{i=2}^{k} \frac{1}{i-1} X_i \mid \mathcal{F}_{k-1} \right] \\
    & = S_{k-1} + \frac{1}{k-1}\E[X_k\mid\mathcal{F}_{k-1}] = S_{k-1},
\end{align*}
and so the process $(S_k)_{k \ge 2}$ forms a martingale.

By Lemma \ref{lem:finitevar} and \eqref{eq:prop1}, both $\E[\mathcal{D}_i^2]$ and $\E[\mathcal{D}_i \mid \mathcal{F}_{i-1}]$ are uniformly bounded. Therefore, $\E[X_i^2]$ is also uniformly bounded for all $i \geq 2$. Furthermore, since $\E[X_i X_j] = 0$ for any $i \neq j$, it  follows that for all $k \geq 2$,
$$\E[S_k^2] = \sum_{i=2}^k \frac{1}{(i-1)^2} \E[X_i^2] < \infty.$$ By the martingale convergence theorem, $S_k$ converges almost surely as $k \to \infty$. By Kronecker's lemma (see \citet[Lemma 1, p.239]{feller2008introduction}),
\[
\lim_{k \to \infty} \frac{1}{k-1} (X_2+ \cdots + X_k) = 0 \quad \text{almost surely.}
\]
Substituting the definition of $X_i$, we have
\begin{equation*} \lim_{k \to \infty}\frac{1}{k-1} \sum_{i=2}^{k}\big( \mathcal{D}_i - \mathbb{E}[\mathcal{D}_i\mid \mathcal{F}_{i-1}] \big) =0 \quad \text{almost surely.} \end{equation*}
Using the bound from \eqref{eq:prop1}, this implies 
\begin{align}\label{eq:bound_prop1}
\limsup_{k \to \infty}\frac{1}{k-1} \sum_{i=2}^{k}  \mathcal{D}_i \leq 1/\eps-\delta < 1/\eps \quad \text{almost surely.}
\end{align} 

To connect this to the action volatility, for each integer $n \ge 1$, let $N_S(n) := \max\{i: \mathcal{T}_i \leq n \}$ denote the cumulative number of sign switches of the public LLR up to time $n$, with the convention that $N_S(n)= 0$ when $n < \mathcal{T}_1$.  Since $a_t = \text{sign}(\ell_{t+1})$, the action switch counts differ from $N_S(n)$ by at most a constant. Moreover, by definition, we can write 
\[
\mathcal{T}_{k} = \mathcal{T}_{1} +\sum_{i=2}^{k} \mathcal{D}_i.
\]
It thus follows from \eqref{eq:bound_prop1} that 
\[
\limsup_{k \to \infty} \frac{\mathcal{T}_{k}}{k} \leq 1/\eps - \delta. 
\]
Since $\mathcal{T}_{N_S(n)} \le n < \mathcal{T}_{N_S(n)+1}$ and $N_S(n) \to \infty$ almost surely, 
\begin{align*}
\liminf_{n \to \infty} \mathcal{Q}_{a}(n) = \liminf_{n \to \infty} \frac{N_S(n)}{n} & \ge  \liminf_{n \to \infty} \frac{N_S(n)}{\mathcal{T}_{N_S(n)+1}} \\
& \geq \frac{1}{1/\eps - \delta} > \eps = \lim_{n \to \infty} \mathcal{Q}_{\theta}(n) \quad \text{almost surely.}
\end{align*}

\medskip

\subsection*{Part (ii)}
We now prove that $\lim_{n\to \infty} \mathcal{Q}_s(n) > \limsup_{n\to \infty} \mathcal{Q}_a(n)$ almost surely. Dividing \eqref{eq:counting} in Lemma \ref{lem:action_signal} by $n$ yields
\[
\mathcal{Q}_a(n) \le \mathcal{Q}_s(n) - \frac{W_n}{n} + \frac{1}{n},
\]
so it suffices to show that
\begin{equation}\label{eq:st_pos_waste}
\liminf_{n\to\infty} \frac{W_n}{n} > 0 \quad \text{almost surely}.
\end{equation}
Recall that $\mathcal{C}_1 < \mathcal{C}_2 < \cdots$ are the successive times at which the public LLR enters the cascade region, and $\mathcal{J}_i = \mathcal{C}_i - \mathcal{C}_{i-1}$ (with $\mathcal{C}_0 = 0$) are the durations between cascade entries. Let $\mathcal{H}_i = \sigma(\mathcal{J}_1, \ldots, \mathcal{J}_i)$ and define the martingale differences $$Y_i := \mathcal{J}_i - \mathbb{E}[\mathcal{J}_i \mid \mathcal{H}_{i-1}]\quad \text{for all}~ i \ge 2.$$

It thus follows from Proposition \ref{prop:Jbound} and the same argument that led to \eqref{eq:bound_prop1} that 
\[
\limsup_{k\to\infty}\frac{1}{k}\sum_{i=1}^k \mathcal{J}_i \le \lfloor K(\alpha,\eps)\rfloor + N 2^{N}\quad\text{almost surely.}
\]
Since $\sum_{i=1}^k \mathcal{J}_i = \mathcal{C}_k$, it follows that
\begin{equation}\label{eq:cascadefreq}
\liminf_{k\to\infty}\frac{k}{\mathcal{C}_k} = \liminf_{n\to\infty}\frac{N_{C}(n)}{n} \ge \frac{1}{\lfloor K(\alpha,\eps)\rfloor + N 2^{N}}\quad\text{almost surely},
\end{equation}
where $N_{C}(n) :=\max\{i:\ \mathcal{C}_i\le n\}$ denotes the number of cascade entries by time $n$, again with the convention $N_{C}(n) = 0$ when $n < \mathcal{C}_1$. 

Next, we show that each cascade entry produces a wasted signal switch with uniformly positive probability. For each $i\ge 1$, define the event
\[
F_i := \{s_{\mathcal{C}_i}=-a_{\mathcal{C}_i}, s_{\mathcal{C}_i+1}=a_{\mathcal{C}_i}\}.
\]
On this event, clearly, $s_{\mathcal{C}_i}\neq s_{\mathcal{C}_i+1}$. By definition, at time $\mathcal{C}_i$, the public LLR is in the cascade region, and so $a_{\mathcal{C}_i}$ depends only on public action history but not on $s_{\mathcal{C}_i}$. The condition $s_{\mathcal{C}_i+1}=a_{\mathcal{C}_i}$ then ensures that $a_{\mathcal{C}_i+1}=a_{\mathcal{C}_i}$ regardless of the region at time $\mathcal{C}_i+1$. Thus, every occurrence of $F_i$ contributes to the wasted switch count $W_n$, i.e. 
\begin{equation}\label{eq:WgeF}
W_n \ge \sum_{i=1}^{N_{C}(n)} \mathbbm{1}(F_i).
\end{equation}

We now lower bound the probability of $F_i$. Let $\mathcal{G}_{i} = \sigma(s_1, \dots, s_{\mathcal{C}_i -1})$ and note that $a_{\mathcal{C}_i}$ is $\mathcal{G}_{i}$-measurable. Since the state is Markov, regardless of the history, $$\Prob[\theta_{\mathcal{C}_{i}+1} = \theta_{\mathcal{C}_{i}} \mid \mathcal{G}_{i}] = 1-\eps.$$ 
Conditional on this event and the common state, the signals $s_{\mathcal{C}_i}, s_{\mathcal{C}_i+1}$ are independent draws with accuracy $\alpha$. Hence, conditional on $\mathcal{G}_{i}$, we have
\[
\mathbb{P}[F_i \mid \theta_{\mathcal{C}_i+1}=\theta_{\mathcal{C}_i}, \mathcal{G}_i] = \alpha(1-\alpha),
\]
and hence by the law of total probability,
\[
\mathbb{P}[F_i \mid \mathcal{G}_{i}] \ge (1-\eps)\alpha(1-\alpha).
\]
Since \(\mathcal C_{i+1}\geq\mathcal C_i+2\), the indicator \(\mathbbm{1}(F_i)\) is \(\mathcal G_{i+1}\)-measurable. Thus, the differences \((\mathbbm{1}(F_i)-\mathbb P[F_i\mid\mathcal G_i])_{i}\) form a bounded martingale difference sequence.  Applying the same martingale argument as in part (i), we have
\[
\liminf_{k\to\infty}\frac{1}{k}\sum_{i=1}^k \mathbbm{1}(F_i) \ge (1-\eps)\alpha(1-\alpha)\quad\text{almost surely.}
\]
Combining this with \eqref{eq:cascadefreq} and \eqref{eq:WgeF}, and using $N_{C}(n) \to 
\infty$ almost surely, we obtain
\begin{align*}
    \liminf_{n\to\infty}\frac{W_n}{n} 
    \ge \liminf_{n\to\infty} \frac{1}{n} \sum_{i=1}^{N_{C}(n)} \mathbbm{1}(F_i) &\ge  \left( \liminf_{n\to\infty} \frac{N_{C}(n)}{n} \right) \cdot \left( \liminf_{k\to\infty} \frac{1}{k} \sum_{i=1}^{k} \mathbbm{1}(F_i) \right) \\
    &\ge \frac{(1-\eps)\alpha(1-\alpha)}{\lfloor K(\alpha,\eps)\rfloor+N2^{N}} > 0 \quad\text{almost surely,}
\end{align*}
as required by \eqref{eq:st_pos_waste}. This concludes the proof of Theorem \ref{cor:exptime}.

\begin{proof}[Proof of Proposition \ref{prop:new1}]
We first bound the action volatility in the public-action setting from above by showing that two supporting actions at successive visits to the learning region initiate a long cascade. By symmetry, suppose $0<\ell_t < c_\alpha$ and $a_t = \text{sign}(\ell_t)$. At the first period after $t$ in which the public LLR is in the learning region, it is at least \(G_\eps(c_\alpha)\): if no cascade begins, monotonicity gives  \(\ell_{t+1}  =G_\eps(\ell_t + c_\alpha)>G_\eps(c_\alpha)\); otherwise, the final cascade update applies \(G_\eps\) to $\ell \ge c_\alpha$, giving the same lower bound. Thus, a second supporting action at this next visit produces a public LLR of at least $G_\eps\!\left(G_\eps(c_\alpha)+c_\alpha\right)$. Let $q^*$ denote the public belief corresponding to this value.

Using \(2q-1=(e^\ell-1)/(e^\ell+1)\), we obtain
\[
\frac{q^* - 1/2}{\alpha - 1/2} = 
\frac{1-2\eps}{2\alpha-1}
\frac{e^{G_\eps(c_\alpha)+c_\alpha}-1}
     {e^{G_\eps(c_\alpha)+c_\alpha}+1} = \frac{(1-2\eps)(2-2\eps)}
     {1+(1-2\eps)(2\alpha-1)^2} \to 2,
\]
as $(\alpha, \eps) \to (1/2, 0)$. During a cascade, from \eqref{eq:cascade_LLR}, we know that the distance of the public belief from $1/2$ contracts by a factor of $1-2\eps$ each period.  Following the same calculation as in the proof of Lemma \ref{lem:cascade_length}, the duration of a cascade starting at $q^*$ is 
\[
h^*= \left\lfloor\frac{\log(\frac{q^*-1/2}{\alpha - 1/2})}{-\log(1-2\eps)} \right\rfloor +1,
\]
since the cascade ends when the belief falls strictly below $\alpha$. Since \(\eps/(-\log(1-2\eps)) \to1/2\), we have $h^* \eps \to (\log 2)/2 > 1/3$ as $(\alpha, \eps) \to (1/2, 0)$. Hence for $\alpha$ sufficiently close to 1/2 and $\eps$ sufficiently small, every such pair initiates a cascade lasting at least \( h^* \ge 1/(3\eps)\) periods. 

By \eqref{eq:max_prob}, a supporting action when $0<|\ell_t| < c_\alpha$ has conditional probability at least 1/2, so two at successive visits have conditional probability at least \(1/4\). Applying a martingale argument similar to that in the proof of Theorem \ref{cor:exptime} shows that at least 1/4 of all adjacent pairs of learning visits yield two supporting actions in the long run. Each initiates a distinct cascade lasting at least \(1/(3\eps)\) periods, so
\[
\limsup_{n\to\infty}\frac1n
\sum_{t=1}^n\mathbbm{1}(|\ell_t|<c_\alpha)
\le \frac{4}{1/(3\eps)}
=12\eps
\qquad\text{almost surely}.
\]
Since action switches occur only in the learning region, with conditional probability at most \(1/2\), the same martingale argument yields
\[
\limsup_{n\to\infty}\mathcal Q_a(n)
\le
\frac12\limsup_{n\to\infty}\frac1n
\sum_{t=1}^n\mathbbm{1}(|\ell_t|<c_\alpha)
\le6\eps
\qquad\text{almost surely}.
\]

We now derive a lower bound on action volatility in the public-signal setting and show that it is strictly greater than $6\eps$ for $\alpha$ sufficiently close to 1/2.  Let $q^S_t := \Prob[\theta_t = 1 \mid s_1, \ldots, s_{t-1}]$ denote the corresponding public belief at time $t$. By the law of iterated expectations and the state transition law, 
\begin{equation}
    \mathbb E[q_{t+1}^S-\tfrac12\mid q_t^S]
=(1-\eps)q_t^S+\eps(1-q_t^S)-\tfrac12
=(1-2\eps)(q_t^S-\tfrac12).
\end{equation}
Since each agent chooses the action most likely to match the state, we have
\(a_t^S(q_{t+1}^S-1/2)=|q_{t+1}^S-1/2|\) for all $t \ge 1$. Thus, for \(t\ge2\), 
\begin{equation} \label{eq:switch_id}
    |q_{t+1}^S-\tfrac12|
-a_{t-1}^S(q_{t+1}^S-\tfrac12)
=2|q_{t+1}^S-\tfrac12|\,
\mathbbm{1}(a_t^S \ne a_{t-1}^S).
\end{equation}
By the law of iterated expectations, we also have
\[
\E[a_{t-1}^S (q_{t+1}^S-\frac12) \mid q_t^S] = (1-2\eps)\E[a_{t-1}^S (q_{t}^S-\frac12) \mid q_t^S] = (1-2\eps) |q_{t}^S-1/2|.
\]  
Moreover, on an action switch, the updating rule implies \(|\ell_{t+1}^S|\le G_\eps(c_\alpha)\). Using \(2q-1=(e^\ell-1)/(e^\ell+1)\), we obtain 
$$|q_{t+1}^S-\tfrac12|
\le
\frac{e^{G_\eps(c_\alpha)}-1}
{2(e^{G_\eps(c_\alpha)}+1)}
=(1-2\eps)(\alpha-\tfrac12).$$
Taking conditional expectations in \eqref{eq:switch_id} thus gives
\begin{align*}
  \E[|q_{t+1}^S-\tfrac12|\mid q_t^S]
-(1-2\eps)|q_t^S-\tfrac12|& =
2\E[|q_{t+1}^S-\tfrac12|\,
\mathbbm{1}(a_t^S\ne a_{t-1}^S)\mid q_t^S]\\
&\le
(1-2\eps)(2\alpha-1)
\Prob(a_t^S\ne a_{t-1}^S\mid q_t^S).
\end{align*}
Summing over $t$, the same martingale argument yields 
\[
\liminf_{n\to\infty}\mathcal Q_a^S(n)
\ge
\frac{2\eps}{(1-2\eps)(2\alpha-1)}
\liminf_{n\to\infty}\frac1n
\sum_{t=1}^n |q_t^S-\tfrac12|
\quad\text{almost surely}.
\]

To bound this average distance from below, a straightforward calculation using Bayes’ rule and \(\eps\le(2\alpha-1)^2\) gives
\[
\mathbb E[(q_{t+1}^S-\tfrac12)^2\mid q_t^S]
-(q_t^S-\tfrac12)^2 \ge
(2\alpha-1)^2
\bigl[\frac{(1-2\eps)^2}{4}-6(q_t^S-\tfrac12)^2\bigr].
\]
Summing this inequality and applying the same martingale argument yields 
\[
\liminf_{n\to\infty}\frac1n
\sum_{t=1}^n(q_t^S-\tfrac12)^2
\ge\frac{(1-2\eps)^2}{24}
\qquad\text{almost surely}.
\]
Since \(|q_t^S-\tfrac12|\ge2(q_t^S-\tfrac12)^2\), it follows from above that 
\[
\liminf_{n\to\infty}\mathcal Q_a^S(n)
\ge \frac{2\eps}{(1-2\eps)(2\alpha-1)} \frac{(1-2\eps)^2}{12}= \frac{\eps(1-2\eps)}{6(2\alpha-1)} \quad \text{almost surely}.
\]
This lower bound exceeds $6\eps$ if and only if $1-2\eps > 36(2\alpha-1)$. Under $\eps \leq (2\alpha -1)^2$, $1-2\eps$ is at least $1-2(2\alpha-1)^2$, which tends to one as $\alpha \to 1/2$, while $36(2\alpha-1)$ tends to zero. Hence, we can choose $\delta >0$ sufficiently small so that this strict inequality holds and the upper bound for the public-action setting also holds whenever $\alpha < 1/2+\delta$ and $\eps \leq (2\alpha-1)^2$. This concludes the proof. 
\end{proof}

\begin{proof}[Proof of Proposition \ref{prop:new2}]
Fix $\gamma \in (1/2,1)$ and let $\eps = \gamma(1-\alpha)$. Again, we first bound action volatility from below in the public-action setting and then from above in the public-signal setting. Substituting $\eps = \gamma(1-\alpha)$ into the upper bound on cascade length $K(\alpha, \eps)$ gives
    \[
    K(\alpha, \gamma(1-\alpha))= \frac{\log (1-2\alpha(1-\alpha))}
{\log (1-2\gamma(1-\alpha))}\to \frac1\gamma\in(1,2), \quad \text{as}~\alpha \to 1.
    \]
    Hence, for \(\alpha\) sufficiently close to one, every cascade lasts one period, so the public belief is in the learning region in at least \(\lfloor n/2\rfloor\) of the first \(n\) periods. In the learning region, $a_t = s_t$, and \eqref{eq:max_prob} implies that the conditional probability of an action switch is at least $2\alpha(1-\alpha)$. Applying the same
martingale argument as in the proof of Theorem \ref{cor:exptime} yields
    \[
\liminf_{n\to\infty}\mathcal Q_a(n)
\ge
2\alpha(1-\alpha)
\lim_{n\to\infty}\frac{\lfloor n/2\rfloor}{n}
=
\alpha(1-\alpha)
\qquad\text{almost surely}.
\]

Next, we bound action volatility from above in the public-signal setting. Let $\ell_t^S$ denote the public LLR, which follows  \eqref{eq:learning_LLR} in every period:
\[
\ell_{t+1}^S = G_\eps(\ell_t^S + c_\alpha s_t).
\]
Note that the state transition law implies that $|\ell_{t+1}^S| \leq \log \frac{1-\eps}{\eps}$. By monotonicity of $G_\eps$, the lowest LLR after three positive signals is obtained by starting at the lower bound $\ell_t^S=\log(\eps/(1-\eps))$. Iterating the updating rule three times and using \(\varepsilon=\gamma(1-\alpha)\) gives
\[
\ell^{S}_{t+3}-c_\alpha = \log
\frac{
\gamma(1-\varepsilon)
\left[(1-\varepsilon)^2
-2(1-2\varepsilon)\alpha(1-\alpha)\right]
}{
\alpha\left[
\gamma^2(1-\varepsilon)^2
+(1-2\varepsilon)(1-\alpha)(1-\gamma)
\left(1+\gamma(2\alpha-1)\right)
\right]
},
\]
which converges to $-\log\gamma>0$ as $\alpha\to1$. Hence, for \(\alpha\) sufficiently close to one, \(\ell_{t+3}^S>c_\alpha\), so the resulting public LLR lies outside the learning region. By symmetry, three negative signals put it below \(-c_\alpha\). Consequently, whenever the public belief is in the learning region, the preceding three signals cannot all be identical. Thus, for \(t\ge4\), 
\[
\mathbbm{1}(|\ell_t^S|<c_\alpha) \leq \mathbbm{1}(s_{t-3}\ne s_{t-2})
+\mathbbm{1}(s_{t-2}\ne s_{t-1}).
\]
Summing over $t$ and using \eqref{eq:signal_switches}, we obtain
\begin{align*}
    \limsup_{n\to\infty}\frac1n\sum_{t=1}^n \mathbbm{1}(|\ell_t^S|<c_\alpha)
 \le
2\bigl[\eps+2(1-2\eps)\alpha(1-\alpha)\bigr] \leq (4+2\gamma)(1-\alpha)  \quad\text{almost surely},
\end{align*}
where the last inequality uses \(\eps=\gamma(1-\alpha)\) and \((1-2\eps)\alpha\le1\). 

We next bound the frequency of action switches for which the current action does not match the state. 
Such a switch requires the public belief to be in the learning region and the current signal to be incorrect:
\[
\mathbbm{1}(a_t^S\ne a_{t-1}^S,\ a_t^S\ne\theta_t)
\le
\mathbbm{1}(|\ell_t^S|<c_\alpha)
\mathbbm{1}(s_t\ne\theta_t).
\]
Since $\mathbbm{1}(|\ell_t^S|<c_\alpha)$ is determined by past
signals and the current signal has conditional probability 
$1-\alpha$ of being incorrect given past states and signals, the same martingale argument used before gives
\[
\frac1n\sum_{t=1}^n
\mathbbm{1}(|\ell_t^S|<c_\alpha)
[\mathbbm{1}(s_t\ne\theta_t) - (1-\alpha)]
\rightarrow 0
\qquad\text{almost surely}.
\]
It thus follows from above that
\begin{align*}
\limsup_{n\to\infty}\frac1n
\sum_{t=2}^n
\mathbbm{1}(a_t^S\ne a_{t-1}^S,\ a_t^S\ne\theta_t)&\le
(1-\alpha)
\limsup_{n\to\infty}\frac1n
\sum_{t=1}^n\mathbbm{1}(|\ell_t^S|<c_\alpha)\\
& \le
(4+2\gamma)(1-\alpha)^2
\qquad\text{almost surely}.
\end{align*}
Finally, note that for every \(t\ge2\), we can write
\begin{align*}
\mathbbm{1}(a_t^S\ne a_{t-1}^S)
& = 2\mathbbm{1}(a_t^S\ne a_{t-1}^S, a_t^S\ne\theta_t)+\mathbbm{1}(a_{t-1}^S\ne\theta_{t})
-\mathbbm{1}(a_t^S\ne\theta_t) \\
& \le 2\mathbbm{1}(a_t^S\ne a_{t-1}^S, a_t^S\ne\theta_t)+  \mathbbm{1}(\theta_t \neq \theta_{t-1}) +
\mathbbm{1}(a_{t-1}^S\ne\theta_{t-1})
-\mathbbm{1}(a_t^S\ne\theta_t),
\end{align*}
where the inequality follows since \(a_{t-1}^S\ne\theta_t\) implies either \(a_{t-1}^S\ne\theta_{t-1}\) or \(\theta_t\ne\theta_{t-1}\). Summing over $t$, we obtain
\begin{align*}
\sum_{t=2}^n\mathbbm{1}(a_t^S\ne a_{t-1}^S)
\le
\sum_{t=2}^n\mathbbm{1}(\theta_t\ne\theta_{t-1})& +
2\sum_{t=2}^n
\mathbbm{1}(a_t^S\ne a_{t-1}^S,\ a_t^S\ne\theta_t)\\
& +
\mathbbm{1}(a_1^S\ne\theta_1)
-\mathbbm{1}(a_n^S\ne\theta_n).
\end{align*}
After dividing by \(n\), the final two terms vanish as $n \to \infty$. Using \(\mathcal Q_\theta(n)\to\eps = \gamma(1-\alpha)\) almost surely and combining the preceding bounds yields 
\[
\limsup_{n\to\infty}\mathcal Q_a^S(n)
\le
\gamma(1-\alpha)+(8+4\gamma)(1-\alpha)^2
\qquad\text{almost surely}.
\]
This upper bound is
strictly below the public-action lower bound $\alpha(1-\alpha)$ whenever 
$1-\alpha<\frac{1-\gamma}{9+4\gamma}$. Choose $\delta>0$ sufficiently small so that this inequality and the preceding estimates hold whenever $\alpha>1-\delta$. This concludes the proof.
\end{proof}

\bibliographystyle{ecta}
\bibliography{reference.bib}

\end{document}